\documentclass[sigconf,10pt]{acmart}

\usepackage[english]{babel}

\renewcommand\footnotetextcopyrightpermission[1]{} 
\setcopyright{none}

\usepackage{amsmath}
\usepackage{amssymb}
\usepackage{graphicx} 
\usepackage{subfigure}
\usepackage{url}
\usepackage{booktabs} 
\usepackage{array} 
\usepackage{verbatim} 
\usepackage{caption}
\usepackage{enumitem}
\newtheorem{mylemma}{Lemma}

\usepackage{color}

\newcommand{\eq}[1]{Eq.~\eqref{#1}}
\usepackage{soul}
\setstcolor{red}
\newcommand{\myitem}[1]{\vspace{0.25\baselineskip}\noindent\textbf{#1}}

\definecolor{orange}{rgb}{1,0.5,0}

\usepackage{xspace}
\newcommand{\ourAlgo}[0]{\texttt{CABaRet}\xspace}
\newcommand{\secref}[1]{Sec.~\ref{#1}}

\usepackage{algorithm}
\usepackage{algorithmicx}
\usepackage{algpseudocode}

\begin{document}

\title{CABaRet: Leveraging Recommendation Systems\\ for Mobile Edge Caching}


\author{Savvas Kastanakis}
\affiliation{%
  \institution{University of Crete and FORTH, Greece}
}
\email{kastan@csd.uoc.gr}

\author{Pavlos Sermpezis}
\affiliation{%
  \institution{FORTH, Greece}
}
\email{sermpezis@ics.forth.gr}

\author{Vasileios Kotronis}
\affiliation{%
  \institution{FORTH, Greece}
}
\email{vkotronis@ics.forth.gr}

\author{Xenofontas Dimitropoulos}
\affiliation{%
  \institution{University of Crete and FORTH, Greece}
}
\email{fontas@ics.forth.gr}

\begin{abstract}
Joint caching and recommendation has been recently proposed for increasing the efficiency of mobile edge caching. While previous works assume collaboration between mobile network operators and content providers (who control the recommendation systems), this might be challenging in today's economic ecosystem, with existing protocols and architectures. In this paper, we propose an approach that enables cache-aware recommendations without requiring a network and content provider collaboration. We leverage information provided publicly by the recommendation system, and build a system that provides cache-friendly and high-quality recommendations. We apply our approach to the YouTube service, and conduct measurements on YouTube video recommendations and experiments with video requests, to evaluate the potential gains in the cache hit ratio. Finally, we analytically study the problem of caching optimization under our approach. Our results show that significant caching gains can be achieved in practice; 8 to 10 times increase in the cache hit ratio from cache-aware recommendations, and an extra 2 times increase from caching optimization.
\end{abstract}




\maketitle

\section{Introduction}
\label{sec:intro}
Mobile Edge Caching (MEC) is one of the key technologies for 5G networks~\cite{MEC:white} that can reduce latency of service delivery and offload traffic from backhaul links. In MEC, caches are located at the edge of the mobile network (e.g., base stations), and thus have limited capacity and serve small --and frequently changing-- user populations~\cite{Paschos-infocom2016}. These factors, despite the advances in caching policies~\cite{Paschos-infocom2016} or delivery techniques~\cite{femto}, limit the possible gains from MEC: capacity is a tiny fraction of today's content catalogs, and traffic is highly variable; hence, a large number of user requests is for non-cached contents, i.e., not served in the edge.

A recently proposed solution for increasing the efficiency of MEC is jointly caching and recommending content~\cite{sermpezis-sch-globecom,chatzieleftheriou2017caching,giannakas-wowmom-2018}. Recommendation Systems (RS) are integrated in many popular services (e.g., YouTube, Netflix) and significantly affect the user demand~\cite{RecImpact-IMC10, gomez2016netflix}. Therefore, steering recommendations towards cached contents, can significantly increase the cache hit ratio, even with small caches or populations.

However, joint caching and recommendation requires collaboration between network operators and Content Providers (CPs). This might be challenging, due to the different scope of these entities, and the constraints of current network protocols and architectures. For example, CPs encrypt traffic (e.g., https) and do not typically share user-related information~\cite{leguay2017cryptocache}. 

To bridge this gap, we propose an approach that is applicable in today's networks: the network operator leverages the information made available by the RS, and, based on this, provides \textit{independently of the CP} high-quality and cache-friendly recommendations that increase the efficiency of MEC. Specifically, we consider the YouTube service, and design a system/application that (i) obtains video relations from the YouTube API, based on which (ii) it builds extended lists of directly and indirectly related videos, and (iii) carefully steers initial recommendations --and thus user demand-- towards cached videos. These operations can take place without any tight collaboration with the CP, thus facilitating the application of joint caching and recommendation approaches by network operators (or other entities), without any need for modifications in architectures or protocols.

Our contributions are summarized as follows:
\begin{itemize}[leftmargin=*]
\item We propose an approach that enables joint caching and recommendation, without requiring collaboration between network operators and CPs (\secref{sec:overview}). 
\item We design an algorithm (named \ourAlgo) that leverages available information provided by a RS, and returns cache-aware recommendations (\secref{sec:recommendation}).
\item We perform extensive measurements over the YouTube service. Our results show that significant caching gains can be achieved in practice; even in conservative scenarios, our approach increases the cache hit ratio by a factor of 8 to 10 (\secref{sec:measurements}).
\item We analytically study the problem of caching optimization under recommendations from \ourAlgo, and propose an approximation algorithm. We show that when caching is controlled by the network operator, an extra 2 times increase in the cache hit ratio can be achieved (\secref{sec:joint}).  
\end{itemize}
Finally, while in this paper we focus on the YouTube service, which provides a public API, our approach can be extended to other video/radio services (e.g., Neflix, Vimeo, Spotify, Pandora) as well, e.g., using offline crawling processes (in case APIs are not available) for discovering content relations.

\section{System Overview}
\label{sec:overview}

The proposed approach can be implemented in a lightweight system/application that runs on mobile devices, and is triggered either by the network operator or by the user. The system is composed of the \textit{user interface} (UI), the \textit{back-end}, and the \textit{recommendation module}, as depicted in Fig.~\ref{fig:system-overview}.

\myitem{User Interface (UI).} The UI resembles the original content service UI. For instance, in the YouTube case, the UI contains a search bar, a video player, and a list of related videos\footnote{In the remainder we use the terms \textit{content} and \textit{video} interchangeably.}. The users search, browse, and watch videos through the UI.

\myitem{Back-end.} The back-end is responsible for (i) retrieving the list of cached video IDs (e.g., in the form of a text file), and (ii) streaming videos to the UI. Depending on the scenario, the list of cached video IDs can be already known to the network operator, e.g., in the case of network-controlled caching. Alternatively, they can be requested from the content provider directly\footnote{Note that this refers to an \textit{aggregate} list of the IDs of all the videos stored in a cache, without containing any user information or violating privacy policies.}, or discovered through offline network measurements (e.g., latency~\cite{cache-centric-video-recommendation}, or DNS resolution~\cite{adhikari2011you}) by the network operator. The video requested by the user is delivered  / streamed from the CP's (e.g., YouTube's) origin server, or the CP's cache, or an edge cache. In case the caches are controlled by the content provider (which is the most prevalent scenario today), the user-service communication can be encrypted (e.g., https requests directly to YouTube) and remain transparent to the network operator.


\myitem{Recommendation Module.} The recommendation module is triggered upon each content request, and (i) receives as input the video $v$ that the user currently watches, (ii) retrieves from the YouTube API a list of video IDs directly/indirectly related to $v$, (iii) extracts from the back-end the list of cached video IDs, and (iv) builds a list of \textit{related and cached} video IDs and recommends it to the user, according to the cache-aware recommendation algorithm of \secref{sec:recommendation}. This recommendation process is lightweight and can return the list of recommendations very fast (e.g., $\sim$1sec. in our prototype), without affecting the user experience.

\begin{figure}
\centering
\includegraphics[width=1\linewidth]{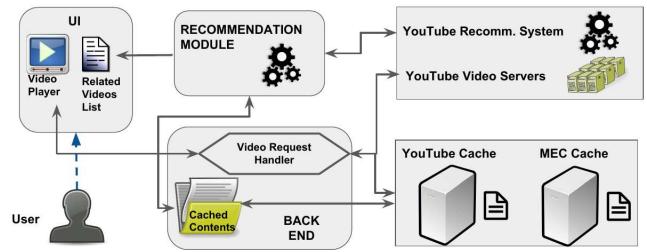}
\caption{System overview.}
\label{fig:system-overview}
\end{figure}

\section{Cache-Aware Recommendations}
\label{sec:recommendation}

In existing approaches, e.g.,~\cite{sermpezis-sch-globecom,chatzieleftheriou2017caching,giannakas-wowmom-2018}, the ``most related'' contents that are also cached, are recommended to users. 
However, this requires the system to be aware of the \textit{content relations} (e.g., similarity scores, user preferences/history, trending videos), i.e., information owned by the content provider. Such data are unlikely to be disclosed to third parties, due to privacy and/or economic reasons (e.g.,~advertising). 

In our approach, the system leverages information about content relations that is made publicly available by the RS of the content service (i.e., YouTube in this paper). In particular, when a user watches a video $v$, the system requests from the YouTube API a list of video IDs $\mathcal{L}$ related to $v$, i.e., the videos that YouTube would recommend to the user. Then it requests the related video IDs for every video in $\mathcal{L}$ and adds them in the end of $\mathcal{L}$, and so on, in a Breadth-First Search (BFS) manner. In the end of the process, the list $\mathcal{L}$ contains IDs of videos directly and indirectly related to $v$, from which the top $N$ cached and/or highly related to $v$ videos are finally recommended to the user. The list $\mathcal{L}$ is (i) much larger than the list of videos recommended by YouTube, and thus it is more probable to contain cached videos that are related to $v$, and (ii) built based on video relations provided by YouTube itself, which satisfies a high quality of recommendations.


We detail our recommendation algorithm (\ourAlgo) in \secref{sec:detailed-algo}, and discuss the related design implications in \secref{sec:tune-algo}.

\begin{figure}
\centering
\includegraphics[width=1\linewidth,height=0.6\linewidth]{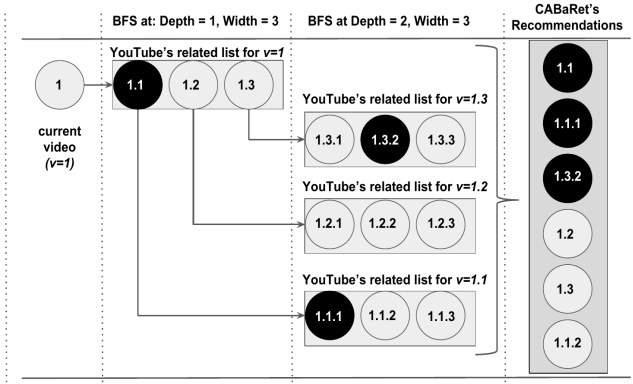}
\caption{CABaRet: example with $D_{BFS}=2$, $W_{BFS}=3$, $N=6$. Cached videos are denoted with black color.}
\label{fig:rec-algo}
\end{figure}

\subsection{The Recommendation Algorithm}\label{sec:detailed-algo}

\myitem{Input.} The recommendation algorithm receives as input:
\begin{itemize}[leftmargin=*]
\item $v$: the video ID (or URL) which is currently watched
\item $N$: the number of videos to be recommended
\item $\mathcal{C}$: the list with the IDs of the cached videos
\item $D_{BFS}$: the depth to which the BFS proceeds
\item $W_{BFS}$: the number of related videos that are requested \textit{per content} from the YouTube API (i.e., the ``width'' of BFS)
\end{itemize}

\myitem{Output.} The recommendation algorithm returns as output:
\begin{itemize}[leftmargin=*]
\item $\mathcal{R}$: ordered list of $N$ video IDs to be recommended.
\end{itemize}
\myitem{Workflow.} \ourAlgo searches for videos related to video $v$ in a BFS manner as follows (\textit{line 1} in Algorithm~\ref{alg:recommendation}). Initially, it requests the $W_{BFS}$ videos related to $v$, and adds them to a list $\mathcal{L}$ in the order they are returned from the YouTube API. For each video in $\mathcal{L}$, it further requests $W_{BFS}$ related videos, as shown in Fig.~\ref{fig:rec-algo}, and adds them in the end of $\mathcal{L}$. It proceeds similarly for the newly added videos, until the depth $D_{BFS}$ is reached; e.g., if $D_{BFS}=2$, then $\mathcal{L}$ contains $W_{BFS}$ video IDs related to $v$, and $W_{BFS}\cdot W_{BFS}$ video IDs related to the related videos~of~$v$.

Then, \ourAlgo searches for video IDs in $\mathcal{L}$ that are also included in the list of cached videos $\mathcal{C}$ 
and adds them to the list of video IDs to be recommended $\mathcal{R}$, until all IDs in $\mathcal{L}$ are explored or the list $\mathcal{R}$ contains $N$ video IDs, whichever comes first (\textit{lines 4--9}).
If after this step, $\mathcal{R}$ contains less than $N$ video IDs, $N-|\mathcal{R}|$ video IDs from the head of the list $\mathcal{L}$ are added to $\mathcal{R}$; these IDs correspond to the top $N-|\mathcal{R}|$ non-cached videos that are directly related to video $v$ (\textit{lines 10--15}).


\begin{algorithm}
\begin{algorithmic}[1]
\caption{\\~\textit{CABaRet}: Cache-Aware \& BFS-related  Recommendations}\label{alg:recommendation}
\Statex {$Input: v, N, \mathcal{C}, D_{BFS}, W_{BFS}$} 
\State $\mathcal{L}\gets BFS(v,D_{BFS},W_{BFS})$ \Comment ordered set of video IDs
\State $\mathcal{R}\gets \emptyset$ \Comment ordered set of video IDs to be recommended
\State $i\gets 1$
\For {$c\in \mathcal{L}$}
	\If {$i\leq N$ and $c\in \mathcal{C}$}
		\State $\mathcal{R}.{append}(c)$
		\State $i\gets i+1$
	\EndIf
\EndFor
\For {$c\in \mathcal{L}\setminus \mathcal{R}$}
	\If {$i\leq N$}
    	\State $\mathcal{R}.{append}(c)$
		\State $i\gets i+1$
    \EndIf
\EndFor

\State $return~~\mathcal{R}$
\end{algorithmic}
\end{algorithm}

\subsection{Implications and Design Choices}\label{sec:tune-algo}

\myitem{High-quality recommendations.} Using the YouTube recommendations ensures strong relations between videos that are directly related to $v$ (i.e., BFS at depth 1). Moreover, while the YouTube RS finds hundreds of videos highly related to $v$, only a subset of them are recommended to the user~\cite{covington2016deep}. The rationale behind \ourAlgo is to explore the related videos that are not communicated to the user. To this end, based on the fact that related videos are similar and have high probability of sharing recommendations (i.e., if video $a$ is related to $b$, and $b$ to $c$, then it is probable that $c$ relates to $a$), \ourAlgo tries to infer these latent video relations through BFS. Hence, videos found by BFS in depths $>1$ are also (indirectly) related to $v$ and probably good recommendations as well.

To further support the above claim, we collect and analyze datasets of related YouTube videos. Specifically, we consider the set of most popular  videos, let $\mathcal{P}$, in a region, and for each $v\in\mathcal{P}$ we perform BFS by requesting the list of related videos (similarly to \textit{line 1} in \ourAlgo). We use as parameters $W_{BFS}=\{10,20,50\}$ and $D_{BFS}=2$, i.e., considering the directly related videos (depth 1) and indirectly related videos with depth 2. We denote as $\mathcal{R}_{1}(v)$ and $\mathcal{R}_{2}(v)$, the set of videos found at the first and second depth of the BFS, respectively. We calculate the fraction of the videos in $\mathcal{R}_{1}(v)$ that are also contained in $\mathcal{R}_{2}(v)$, i.e., $I(v) = \frac{| \mathcal{R}_{1}(v)  \cap \mathcal{R}_{2}(v)|}{|\mathcal{R}_{1}(v)|}$. High values of $I(v)$ indicate a strong similarity of the initial content $v$ with the set of indirectly related contents at depth 2.

Table~\ref{table:I-v} shows the median values of $I(v)$, over the $|\mathcal{P}|=50$ most popular contents in the region of Greece (GR), for different BFS widths. As it can be seen, $I(v)$ is very high for most of the initial videos $v$. For larger values of $W_{BFS}$, $I(v)$ increases, and when we fully exploit the YouTube API capability, i.e., for $W_{BFS}$=50, which is the maximum number of related videos returned by the YouTube API, the median value of $I(v)$ becomes larger than $0.9$. Finally, we measured the $I(v)$ in other regions as well, and observed that even in large (size/population) regions, the $I(v)$ values remain high, e.g., in the United States (US) region, $I(v)$=0.8 for $W_{BFS}$=50.
\begin{table}[h]
\centering
\caption{$I(v)$ vs. $W_{BFS}$ for the region of GR
.}
\vspace{-\baselineskip}
\label{table:I-v}
\begin{tabular}{|cccc|
}
\hline
{$W_{BFS}:$}	&{10}&{20}&{50}
\\
{$I(v):$}	&0.70&0.85&0.92
\\
\hline
\end{tabular}
\vspace{-\baselineskip}
\end{table}




\myitem{Tuning \ourAlgo.} 
The parameters $D_{BFS}, W_{BFS}$ can be tuned to achieve a desired performance, e.g., in terms of probability of recommending a cached or highly related video.


For large $D_{BFS}$, the similarity between $v$ and the videos \textit{at the end of the list} $\mathcal{L}$ is expected to weaken, while for small $D_{BFS}$ the list $\mathcal{L}$ is shorter and it is less probable that a cached content is contained in it. Hence, the parameter $D_{BFS}$ can be used to achieve a trade-off between quality of recommendations (small $D_{BFS}$) and probability of recommending a cached video (large $D_{BFS}$). 
The number of related videos requested per content $W_{BFS}$, can be interpreted similarly to $D_{BFS}$. A small $W_{BFS}$ leads to considering only top recommendations per video, while a large $W_{BFS}$ leads to a larger~list~$\mathcal{L}$. 

\textit{Remark:} YouTube imposes quotas on the API requests per application per day, which prevents API users from setting the parameters $W_{BFS}$ and $D_{BFS}$ to arbitrarily large values. 


In practice, \ourAlgo can be fine-tuned through experimentation with real users, e.g., A/B testing iterations, which is a common approach for tuning recommendation systems~\cite{covington2016deep}.

\section{Measurements and Evaluation}
\label{sec:measurements}
We conduct measurements and experiments over the YouTube service\footnote{Our experiments and use of the YouTube API conform to the YouTube terms of service \url{https://www.youtube.com/static?template=terms}.}, to investigate the performance (in terms of cache hit ratios) of our approach in MEC scenarios. The setup of the scenarios is presented in \secref{sec:experiments-setup}, and the results in~\secref{sec:results}.

\subsection{Setup}
\label{sec:experiments-setup}

\myitem{The YouTube API} provides a number of functions to retrieve information about videos, channels, user ratings, etc. In our measurements, we request 
the following information:
\begin{itemize}
\item the most popular videos in a region (max. 50)
\item the list of related videos (max. 50) for a given video
\end{itemize}
\textit{Remark}: In the remainder, we present results collected during March 2018, for the region of Greece (GR). Nevertheless, our insights hold also in the other regions we tested (e.g., US).
.


\myitem{Caching.} We assume a MEC cache storing the most popular videos in a region. We populate the list of cached contents with the top $C$ 
video IDs returned from the YouTube API.

\myitem{Recommendations.} We consider two classes of scenarios with (i) YouTube and (ii) \ourAlgo recommendations. In both cases, when a user enters the UI, the $50$ most popular videos in her region are recommended to her (as in YouTube's front page). Upon watching a video $v$, a list of $N=20$ videos is recommended to the user; the list is (i) composed of the top $N$ directly related videos returned from the YouTube API (YouTube scenarios), or (ii) generated by \ourAlgo with parameters $N$, $W_{BFS}$ and $D_{BFS}$ (\ourAlgo scenarios).

\myitem{Video Demand.} In each experiment, we assume a user that enters the UI and watches one of the initially recommended (i.e., $50$ most popular) videos at random. Then, the system recommends a list of $N$ videos ($r_{1},r_{2},...,r_{N}$), and the user selects with probability $p_{i}$ to watch $r_{i}$ next. We set the probabilities $p_{i}$ to depend on the order of appearance --and not the content--
and consider \textit{uniform} ($p_{i}=\frac{1}{N}$) and \textit{Zipf} ($p_{i}\sim\frac{1}{i^{\alpha}}$) scenarios; the higher the exponent $\alpha$ of the Zipf distribution, the more preference is given by the user to the top recommendations (user preference to top recommendations has been observed in YouTube traffic~\cite{cache-centric-video-recommendation}).


\subsection{Results}\label{sec:results}
\subsubsection{Single Requests}\label{sec:results-single}~\\
We first consider scenarios of single requests (similarly to~\cite{sermpezis-sch-globecom,chatzieleftheriou2017caching}). In each experiment $i$ ($i=1,...,M$) a user watches one of the top popular videos, let $v_{1}(i)$, and then follows a recommendation and watches a video $v_{2}(i)$. We measure the Cache Hit Ratio (CHR), which we define as the fraction of the \textit{second requests} of a user that are for a cached video (since the first request is always for a cached --top popular-- video):
\begin{equation*}
CHR = \frac{1}{M}\cdot \sum_{i=1}^{M} \mathbb{I}_{v_{2}(i)\in \mathcal{C}}
\end{equation*}
where $\mathbb{I}_{v_{2}(i)\in \mathcal{C}} =1$ if $v_{2}(i)\in \mathcal{C}$ and $0$ otherwise, and $M$ the number of experiments\footnote{We considered all possible experiments on the collected dataset.}.

\myitem{CHR vs. BFS parameters.} Fig.~\ref{fig:chr-vs-bfs-param} shows the CHR achieved by \ourAlgo under various parameters, along with the CHR under regular YouTube recommendations, when caching all the most popular videos ($|\mathcal{C}|$=50). The efficiency of caching significantly increases with \ourAlgo, even when only directly related contents are recommended ($D_{BFS}$=1), i.e., \textit{without loss in recommendation quality}. Just reordering the list of YouTube recommendations (as suggested in~\cite{cache-centric-video-recommendation}), brings gains when $p_{i}$ is not uniformly distributed. However, the added gains by our approach are significantly higher. As expected, the CHR increases for larger $W_{BFS}$ and/or $D_{BFS}$; e.g., \ourAlgo for $W_{BFS}$=50 and $D_{BFS}$=2, achieves 8-10 times 
higher CHR than regular YouTube recommendations. Also, the CHR increases for more skewed $p_{i}$ distributions, since top recommendations are preferred and \ourAlgo places cached contents at the top of the recommendation list.

In experiments concerning the --larger-- US region, the CHR values are lower for both regular YouTube ($<0.5\%$) and \ourAlgo ($1\%-43\%$) recommendations, due to the fact that the top popular videos appear with lower frequency in the related lists. However, the \textit{relative gains} from \ourAlgo are consistent with (or even higher than) the presented results.

\begin{figure}\centering
\begin{minipage}[t]{0.46\linewidth}
\centering
\includegraphics[width=1\columnwidth]{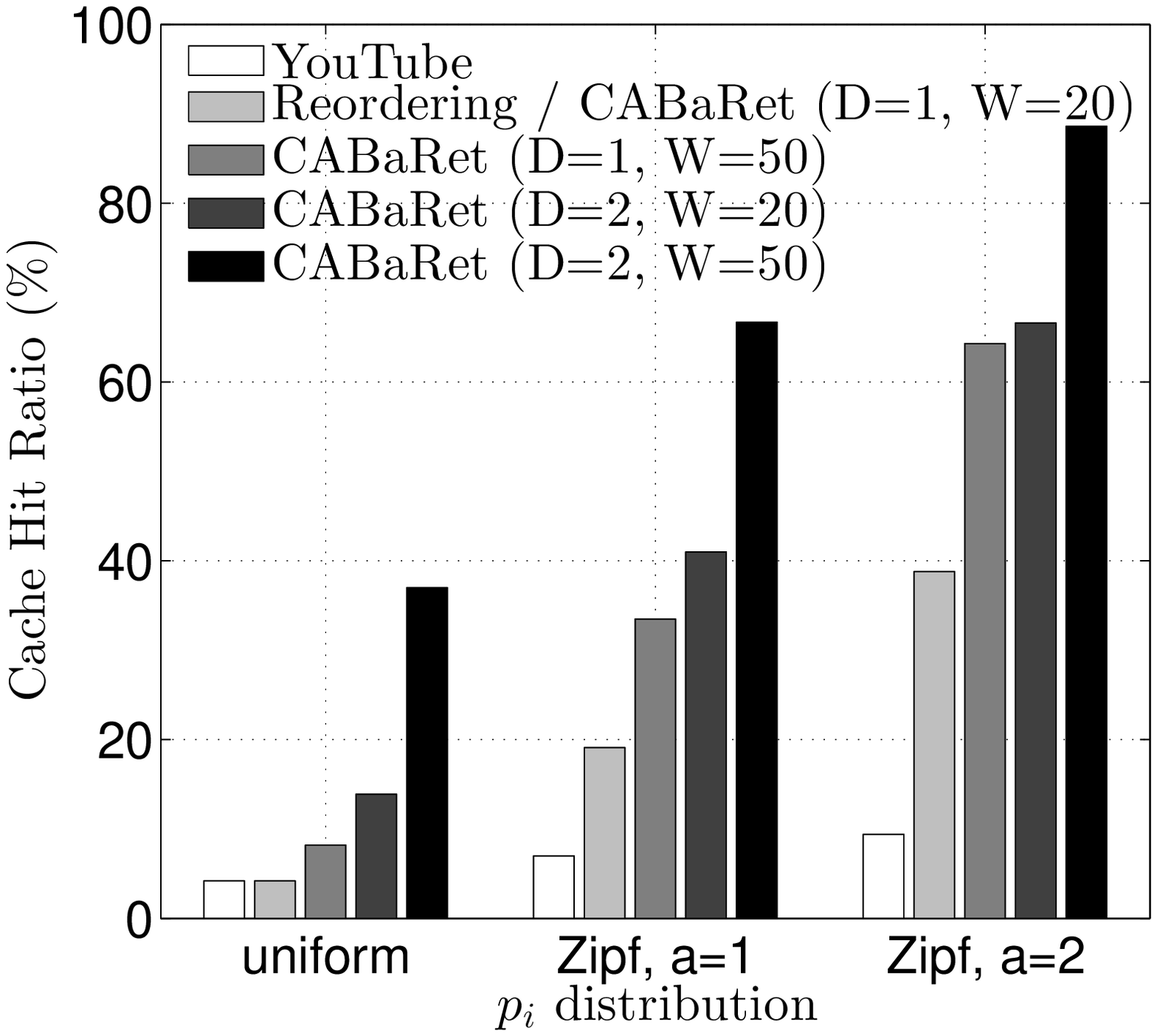}
\caption{CHR under different BFS parameters.
}
\label{fig:chr-vs-bfs-param}
\end{minipage}
\hspace{0.03\linewidth}
\begin{minipage}[t]{0.47\linewidth}
\centering
\includegraphics[width=1\columnwidth]{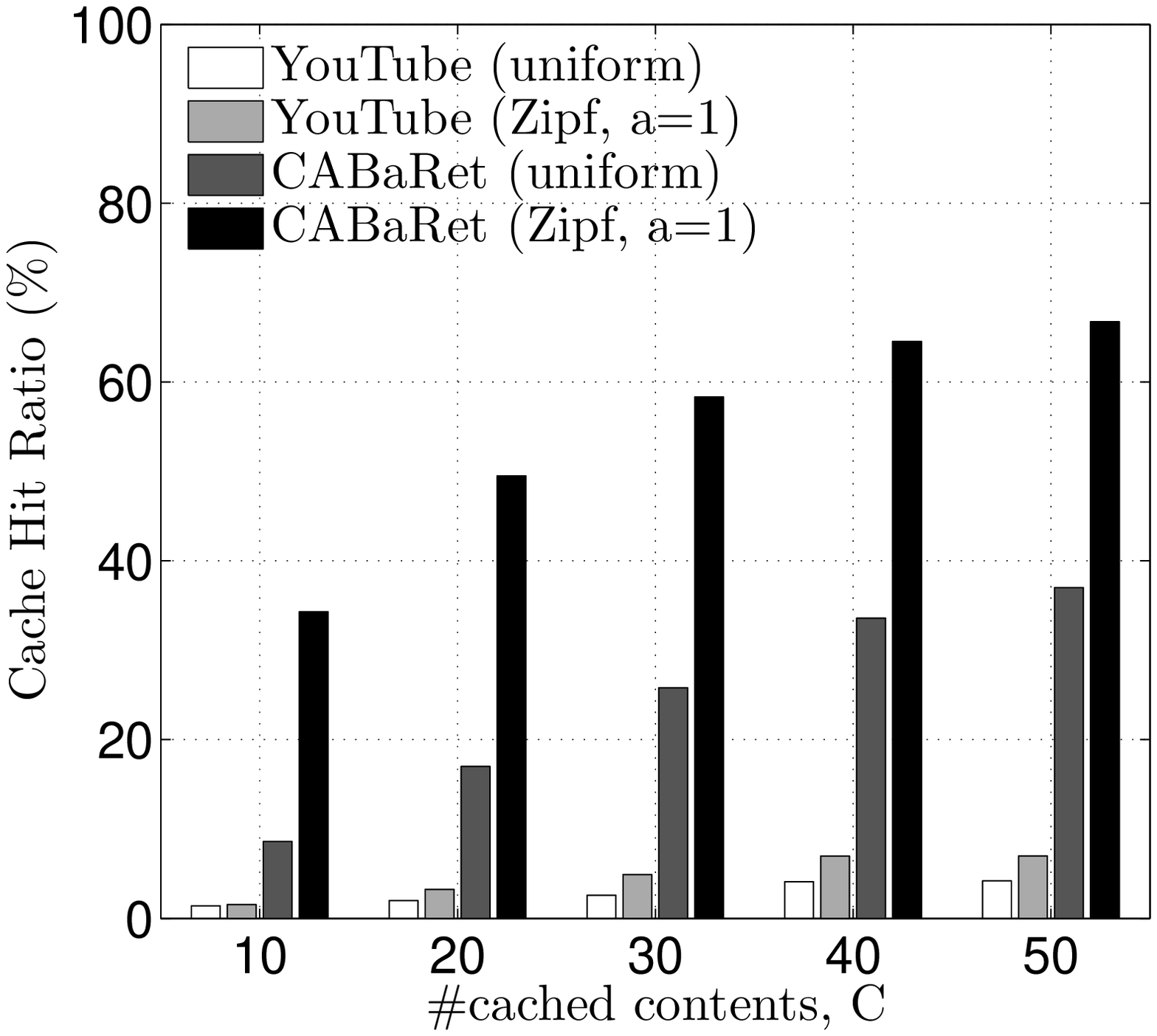}
\caption{CHR vs. \#~cached contents $C$\\ ($W_{BFS}$=50, $D_{BFS}$=2).}
\label{fig:chr-vs-C-d2}
\end{minipage}
\end{figure}

\myitem{CHR vs. number of cached videos.} We further consider scenarios with varying number of cached contents $C = |\mathcal{C}|$. In each scenario, we assume that the $C$ most popular contents are cached. Fig.~\ref{fig:chr-vs-C-d2} 
shows the CHR achieved by \ourAlgo, in comparison to scenarios under regular YouTube recommendations. The results are consistent for all considered values of $C$; the CHR under \ourAlgo is significantly higher than in the YouTube case. Moreover, even when caching a small subset of the most popular videos, \ourAlgo brings significant gains. E.g., by caching $C=10$ out of the $50$ top related contents \ourAlgo increases the CHR from $2\%$ and $3.2\%$ to $17\%$ and $50\%$, for the uniform and Zipf($\alpha$=1) scenarios, respectively.



\subsubsection{Sequential Requests}~\\
We now test the performance of our approach in scenarios where users enter the system and watch a sequence of $K$, $K>2$, videos (similarly to~\cite{giannakas-wowmom-2018}, and in contrast to the previous case, where they watch only two videos, i.e., $K=2$). At each step, the system recommends a list of videos to the user by applying \ourAlgo on the currently watched video. We denote as $v_{k}(i)$ the $k^{th}$ video requested/watched by a user in experiment $i$. We measure the CHR, which is now defined as
\begin{equation*}
CHR = \frac{1}{M}\cdot \sum_{i=1}^{M} \sum_{k=2}^{K} \mathbb{I}_{v_{k}(i)\in \mathcal{C}}
\end{equation*}
where $\mathbb{I}_{v_{k}(i)\in \mathcal{C}} =1$ if $v_{k}(i)\in \mathcal{C}$ and $0$ otherwise, over $M=100$ experiments per scenario.

Moving ``farther'' from the initially requested video (which belongs to the list of most popular and cached videos) through a sequence of requests, we expect the CHR to decrease, due to lower similarity of the requested and cached videos. However, as Fig.~\ref{fig:chr-vs-K-sequential} shows, the decrease in the CHR (under \ourAlgo recommendations) is not large. The CHR remains close to the case of single requests (i.e., for $K$=2 in the x-axis), indicating that our approach performs well even when we are several steps far from the cached videos. In fact, caching more than the top most popular videos appearing on the front page, would further reduce the CHR decrease.

\begin{figure}\centering
\begin{minipage}[t]{0.47\linewidth}
\centering
\includegraphics[width=1\columnwidth]{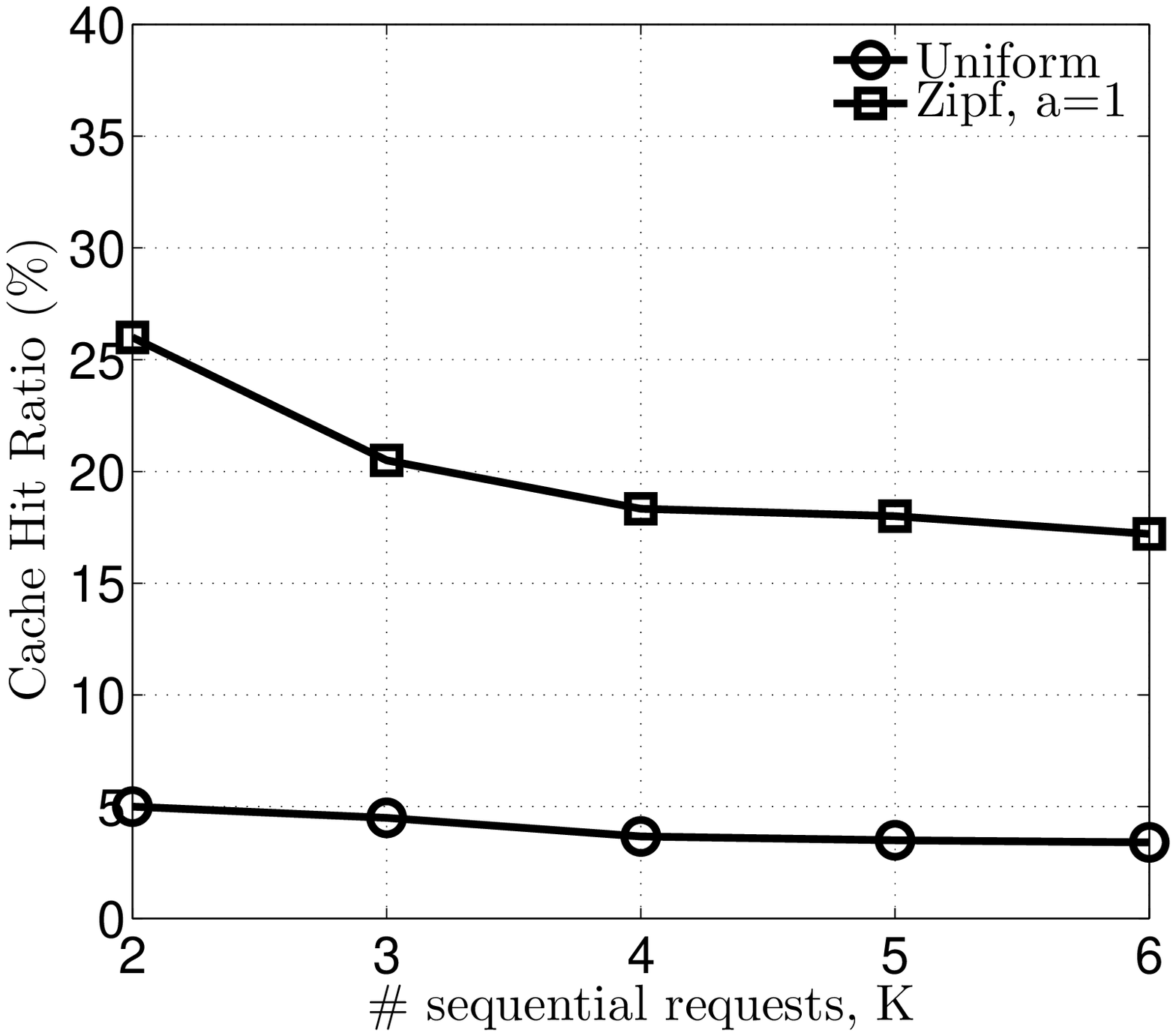}
\caption{CHR vs. \#~requests in sequence $K$ ($C$=20, $W_{BFS}$=20, $D_{BFS}$=2). \textit{Note}: y-axis up to 40\%.}
\label{fig:chr-vs-K-sequential}
\end{minipage}
\hspace{0.03\linewidth}
\begin{minipage}[t]{0.47\linewidth}
\centering
\includegraphics[width=1\columnwidth]{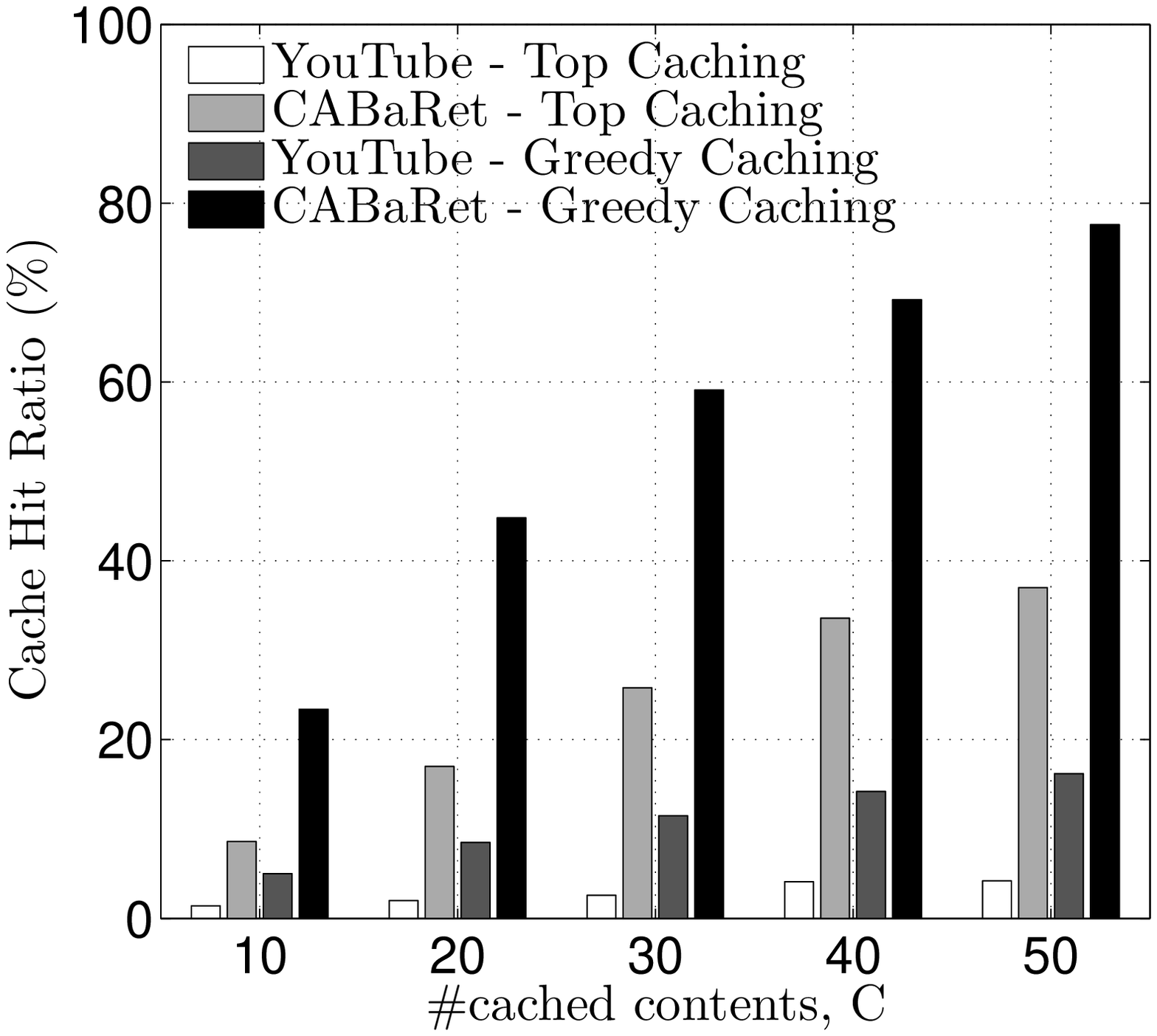}
\caption{CHR vs. \# cached contents ($p_{i}$ \textit{uniform}, $W_{BFS}$=50, $D_{BFS}$=2).}
\label{fig:chr-vs-c-greedy}
\end{minipage}
\end{figure}


\section{Caching Optimization}
\label{sec:joint}
In this section, we extend our study by considering the scenario where \ourAlgo and the caches are controlled by the same entity, e.g., the network operator. Network operator-controlled caching is the most commonly considered scenario in related work for MEC; although in most of today's architectures the caches are actually controlled by the CP. 

In this scenario, the network operator can optimize caching decisions, thus further increasing the efficiency of \ourAlgo recommendations. Note that still there is no need for collaboration between the operator and the CP (e.g., possessing full knowledge of the RS), assumed in previous works~\cite{sermpezis-sch-globecom,chatzieleftheriou2017caching,giannakas-wowmom-2018}.

In the following, we first analytically formulate and study the problem of optimizing the caching policy, and propose an approximation algorithm with provable performance guarantees (\secref{sec:problem-algorithm}). We then evaluate the performance of this \textit{joint caching and recommendation} approach (\secref{sec:results-greedy}).

\subsection{Optimization Problem \&  Algorithm}
\label{sec:problem-algorithm}

Let a content catalog $\mathcal{V}$, $V=|\mathcal{V}|$, and a content popularity vector $\mathbf{q} = [q_{1}, ..., q_{V}]^{T}$. Let $\mathcal{L}(v)\subseteq \mathcal{V}$ be the set of contents that are explored by \ourAlgo (at \textit{line 1}) for a content $v\in\mathcal{V}$, and denote $\mathcal{L} = \bigcup_{v\in\mathcal{V}}\mathcal{L}(v)$.

For some set of cached contents $\mathcal{C}\subseteq\mathcal{V}$, and a content $v$, \ourAlgo returns a list of recommendations $\mathcal{R}(v)$ ($|\mathcal{R}(v)|=N$), in which at most $N$ contents $c\in \mathcal{C}\cap\mathcal{L}(v)$ appear at the top of the list. Therefore, CHR can be expressed as
\begin{equation}\label{eq:chr-generic}
CHR(\mathcal{C}) = \sum_{v\in\mathcal{V}}q_{v}\sum_{i=1}^{N(\mathcal{C},v)} p_{i}
\end{equation}
where $N(\mathcal{C},v) = \min\{|\mathcal{C}\cap\mathcal{L}(v)|, N\}$, and $p_{i}$ is the probability for a user to select the $i^{th}$ recommended content.

Then, the problem of optimizing the caching policy (to be jointly used with \ourAlgo), is formulated as follows:
\begin{equation}\label{optim-problem}
max_{\mathcal{C}} ~CHR(\mathcal{C})~~~~~s.t., |\mathcal{C}|\leq C
\end{equation}
where $C$ is the capacity of the --MEC-- cache.
We prove the following for the optimization problem of \eq{optim-problem}.
\begin{mylemma}
The optimization problem of~\eq{optim-problem}: (i) is NP-hard, (ii) cannot be approximated within $1-\frac{1}{e}+o(1)$ in polynomial time, and (iii) has a monotone (non-decreasing) submodular objective function, and is subject to a cardinality constraint.
\end{mylemma}
\begin{proof}
Items (i) and (ii) of the above lemma, are proven by reduction to the \textit{maximum set coverage} problem, and we prove item (iii) using standard methods (see, e.g.,~\cite{femto,sermpezis-sch-globecom}).
\end{proof}

If we design a greedy algorithm that starts from an empty set of cached contents $\mathcal{C}_{g}=\emptyset$, and at each iteration it augments the set $\mathcal{C}_{g}$ (until $|\mathcal{C}_{g}|= C$) as follows:
\begin{equation}\label{eq:greedy-algo}
\mathcal{C}_{g}\leftarrow \mathcal{C}_{g}\cup \arg\max_{v\in\mathcal{V}} CHR(\mathcal{C}_{g}\cup\{v\}),
\end{equation}
then the properties stated in item (iii) satisfy that it holds~\cite{krause2012submodular}
\begin{equation}\label{eq:greedy-bound}
CHR(\mathcal{C}_{g}) \geq \left(1-\frac{1}{e}\right)\cdot CHR(\mathcal{C}^{*})
\end{equation}
where $\mathcal{C}^{*}$ the optimal solution of the problem of~\eq{optim-problem}. 

\textit{Remark}: While \eq{eq:greedy-bound} gives a lower bound for the performance of the greedy algorithm, in practice greedy algorithms have been shown to perform often very close to the optimal.

\subsection{Evaluation of Greedy Caching}\label{sec:results-greedy}

Calculating the CHR from \eq{eq:chr-generic} requires running a BFS (\ourAlgo, \textit{line 1}) and generating the lists $\mathcal{L}(v)$, for every content $v\in \mathcal{V}$. In practice, for scalability reasons, the most popular contents (i.e., with high $q_{i}$) can be considered by the greedy algorithm in the calculation of the objective function \eq{eq:chr-generic}, since those contribute more to the objective function. However, any video in the catalog is still candidate to be cached, e.g., a video with low $q_{i}$ can bring a large increase in the CHR through its association with many popular contents.

In fact, in our experiments, for the calculation of \eq{eq:chr-generic}, we consider only the 50 most popular videos, for which we set $q_{i} = \frac{1}{50}$. Nevertheless, in the different scenarios we tested, only 10\% to 30\% of the cached videos (selected by the greedy algorithm) were also in the top 50 most popular.

In Fig.~\ref{fig:chr-vs-c-greedy}, we compare the achieved CHR when the cache is populated according to the greedy algorithm of \eq{eq:greedy-algo} (\textit{Greedy Caching}), and with the top most popular videos (\textit{Top Caching}).
~\textit{Greedy caching} always outperforms \textit{top caching}, with an increase in the CHR of around a factor of 2 
for uniform video selection (for the Zipf($a$=1) scenarios we tested, the CHR values are even higher, and the relative performance is 1.5 times higher). This clearly demonstrates that the gains from joint recommendation and caching~\cite{sermpezis-sch-globecom,chatzieleftheriou2017caching}, are applicable even in simple practical scenarios (e.g., \ourAlgo \& greedy caching). Finally, while \textit{greedy caching} increases the CHR even with regular YouTube recommendations, the CHR is still less than $50\%$ of the \ourAlgo case with \textit{top caching}. This further stresses the benefits from \ourAlgo's cache-aware recommendations.



\section{Related Work}
\label{sec:related}

In the recent works of~\cite{sermpezis-sch-globecom,chatzieleftheriou2017caching,giannakas-wowmom-2018}, caching and recommendation of contents is considered for mobile networks. \cite{sermpezis-sch-globecom,chatzieleftheriou2017caching} jointly optimize the caching and recommendation policies, while~\cite{giannakas-wowmom-2018} optimizes cache-aware recommendations for sequential requests. Our work is complementary to these works: we propose an approach for deploying joint caching and recommendation in practice, without collaboration between network operators and content providers. Moreover, our extensive measurements over the real YouTube service show that the gains demonstrated in~\cite{sermpezis-sch-globecom,chatzieleftheriou2017caching,giannakas-wowmom-2018} can be achieved even today, using a simple and lightweight system.   

\cite{cache-centric-video-recommendation} studies how to improve the efficiency of YouTube caches, by reordering the list of recommendations. \ourAlgo goes beyond reordering, and leverages more information that is available through the YouTube API; as our results show, this can bring substantial gains in the CHR.


Finally, research on the performance of the YouTube CDN in relation to network parameters~\cite{cache-centric-video-recommendation,adhikari2011you,
mok2018revealing}, indicates that our approach can be relevant to wireline networks as well. For instance, on top of classic load-balancing, \textit{network-aware recommendations} (e.g., taking into account routing policies~\cite{adhikari2011you,mok2018revealing}, inter-domain load balancing~\cite{
mok2018revealing}, cache locations~\cite{adhikari2011you
}) can be used to improve the user QoE (e.g., latency).


\section{Conclusion}
\label{sec:conclusions}
In this paper, we proposed an approach that enables joint caching and recommendation in MEC, without requiring tight collaboration of network and content providers. We considered the YouTube service as our use case, and showed that significant gains can be achieved by leveraging available information from its recommendation system.

The effect of geographical and temporal factors on the gains from cache-aware recommendations, can be further investigated in the future to extend our initial experimental results. 
Moreover, we believe that analytically studying the interplay between content relations and potential gains with \ourAlgo, as well as extending our approach to other online content services, are interesting research directions. 

\section*{Acknowledgements} This work has been funded by the European Research Council grant agreement no. 338402.




\end{document}